\documentclass[10pt,a4paper]{amsart}

\usepackage{amssymb}
\usepackage{amsmath}
\usepackage{verbatim}
\usepackage{enumerate}
\usepackage{color}
\usepackage{graphicx}
\usepackage{xcolor}
\usepackage{amscd}
\usepackage{amsthm}
\usepackage{epsfig}
\usepackage{eulervm}
\usepackage{subfigure}
\usepackage{extarrows}
\usepackage[
    colorlinks=true,
    linkcolor=black,
    urlcolor=blue,
    citecolor=blue
]{hyperref}
\usepackage{tikz}
\usepackage{braket}
\usepackage{ytableau}
\usepackage[numbers]{natbib}
\usetikzlibrary{decorations.pathreplacing}
\usetikzlibrary{calc}

\theoremstyle{definition}
\newtheorem{remark}{Remark}[section]
\newtheorem{example}[remark]{Example}
\newtheorem{lemma}[remark]{Lemma}
\newtheorem{proposition}[remark]{Proposition}
\newtheorem{theorem}[remark]{Theorem}
\newtheorem{definition}[remark]{Definition}
\newtheorem{corollary}[remark]{Corollary}

\begin{document}
%\author{Maxim Gritskov and Savelii Timchenko}
\title[ 
Perturbative anomalies in quantum mechanics
]{Perturbative anomalies in quantum mechanics}

\begin{abstract}
In this work, we propose a cohomological approach to studying perturbative anomalies in quantum mechanics. The Hamiltonian $\hat{H}$ together with the symmetry generator $\hat{S}$ forms a unitary representation of the two-dimensional Abelian Lie algebra $\mathfrak{g}\cong \mathbb{R}^{2}$ on the Hilbert space $V$. We show that perturbations of such a system are related to the first Chevalley-Eilenberg cohomology group $H^{1}_{CE}(\mathbb{R}^{2},\mathfrak{u}(V))$. In turn, the perturbative anomalies of the symmetry $\hat{S}$ are related to the second cohomology group $H^{2}_{CE}(\mathbb{R}^{2},\mathfrak{u}(V))$. 
\end{abstract}

\author{Maxim Gritskov}
\address{Shanghai Institute for Mathematics and Interdisciplinary Sciences, Building 3, 62 Weicheng Road, Yangpu District, 200433, Shanghai, China}
\address{Saint Petersburg State University,
Universitetskaya nab. 7/9, 199034 St. Petersburg, Russia}
\email{m.gritskov@spbu.ru}

\author{Andrey Losev}
\address{Shanghai Institute for Mathematics and Interdisciplinary Sciences, Building 3, 62 Weicheng Road, Yangpu District, 200433, Shanghai, China}
\email{aslosev2@yandex.ru}

\author{Saveliy Timchenko}
\address{Department of Physics, Fudan University, 2005 Songhu Road, 200438, Shanghai, China}
\email{saul.timchenko@gmail.com}

\maketitle

\setcounter{tocdepth}{3} 
\tableofcontents

\allowdisplaybreaks

\section{Introduction: Lie algebra cohomological program}
In this article, we present our initial arguments in favor of the idea that symmetry anomalies in physics are essentially deformation obstructions. One of the classic questions in quantum field theory is what happens to symmetry when the theory is deformed. It turns out that when studying the question of perturbative symmetry breaking, it is sufficient to study the representation of the symmetry algebra on the state space of the initial theory. In this case, the appropriate language for describing deformations of a theory with symmetry is the language of homological algebra. Namely, let us consider a theory with Hamiltonian $H$ and symmetry algebra generated by $S_{\alpha}$ such that $[S_{\alpha},S_{\beta}]=f_{\alpha\beta}^{\gamma}S_{\gamma}$. Then we define the so-called \emph{Chevalley-Elienberg} differential for the quantum system $(H,S_{\alpha})$ \cite{Losev:2023exl}:
\begin{equation}
    d_{CE}=c^{H}H+c^{\alpha}S_{\alpha}-\frac{1}{2}f_{\alpha\beta}^{\gamma}c^{\alpha}c^{\beta}\frac{\partial}{\partial c^{\gamma}}.
\end{equation}
The space of infinitesimal deformations of theories with symmetries is the first cohomology group $H^{1}_{CE}$ of the Chevalley-Eilenberg algebra of this symmetry.

However, not every infinitesimal deformation defines a curve in the space of theories passing through a given point. The first obstruction is the map 
\begin{equation}
    \mu_{2}: H^{1}_{CE}\otimes H^{1}_{CE}\rightarrow H^{2}_{CE}
\end{equation}
In each order of perturbation theory, obstructions $\mu_{n}: (H^{1}_{CE})^{\otimes n} \rightarrow H^{2}_{CE}$ arise. The maps $\mu_{n}$ form a structure of the $L_{\infty}$-algebra on $H^{\bullet}_{CE}$ and, in particular, satisfy the quadratic relations \cite{kontsevich2000deformation}. In particular, if $S$ satisfies the dilation of space, then $\mu_{n}$ can be interpreted as coefficients of the \emph{beta function} and therefore must satisfy quadratic relations coming from the $L_{\infty}$-structure \cite{Losev:2005pu, Gamayun:2023sif}.

In this text, we implement the program described above using a minimal example which we analyze in detail. 

\section{Perturbative anomalies as the cohomological obstructions}
\subsection{Motivation}
Consider a quantum mechanical system with state space $V$, Hamiltonian $\hat{H}$, and symmetry $\hat{S}$. By symmetry $\hat{S}$, we mean a self-adjoint operator that commutes with the Hamiltonian. For simplicity, we will assume that $\hat{H}$ and $\hat{S}$ have discrete spectra, i.e., that $V$ can be decomposed into a direct sum (possibly infinite) of subspaces that are eigenspaces for both $\hat{H}$ and $\hat{S}$ simultaneously.

Now we will perturb the Hamiltonian, adding a new term $\hat{H}\rightarrow \hat{H}+t\,\delta^{(1)}\hat{H}$. Generally speaking, now $[\hat{S}, \hat{H}+t\,\delta^{(1)}\hat{H}]\neq 0$ and it would seem that the symmetry is broken. But is it possible to deform the symmetry generator $\hat{S}\rightarrow \hat{S}+t\,\delta^{(1)}\hat{S}$ so that it becomes the symmetry of the perturbed problem in first order in $t$? The corresponding correction $\delta^{(1)}\hat{S}$ to the generator $\hat{S}$ satisfies the equation:
\begin{equation} \label{eq:linear-condition}
        [\hat{H},\delta^{(1)}\hat{S}]=[\hat{S},\delta^{(1)}\hat{H}].
    \end{equation}

Suppose that we have managed to restore symmetry in the first order of perturbation theory by selecting the appropriate correction $\delta^{(1)}\hat{S}$. But what about the next orders of perturbation theory? Is it possible to deform the system $(\hat{H},\hat{S})$ along the direction $(\delta^{(1)}\hat{H},\delta^{(1)}\hat{S})$ in all orders of perturbation theory? Formally, the problem can be stated as finding a curve $(\hat{H}(t),\hat{S}(t))$ in the space of theories passing at time $t=0$ through a given point $(\hat{H},\hat{S})$ with velocity $(\delta^{(1)}\hat{H},\delta^{(1)}\hat{S})$.

The perturbative solution to the problem consists in finding two formal series
    \begin{eqnarray}
            \notag
            \hat{H}(t)=\hat{H}+t\,\delta^{(1)}\hat{H}+\sum_{n\geq 2}t^{n}\delta^{(n)}\hat{H},\\
            \hat{S}(t)=\hat{S}+t\,\delta^{(1)}\hat{S}+\sum_{n\geq 2}t^{n}\delta^{(n)}\hat{S}
    \end{eqnarray}
    such that $[\hat{H}(t),\hat{S}(t)]=0$.
 In the second order, condition $[\hat{H}(t),\hat{S}(t)]=0$ requires:
    \begin{equation} \label{eq:quadratic-condition}
        [\hat{H},\delta^{(2)}\hat{S}]+[\delta^{(2)}\hat{H},\hat{S}]+[\delta^{(1)}\hat{H},\delta^{(1)}\hat{S}]=0.
    \end{equation}
    It turns out that this equation is not always solvable. Moreover, it turns out that if the equation \eqref{eq:quadratic-condition} is solvable, then the equations in all higher orders of perturbation theory will also be solvable. To demonstrate this, we will consider this problem from the perspective of a cohomological approach to deformations.

    For this approach, it is more natural to study not the deformation of the two Hermitian operators system, but the deformation of the two anti-Hermitian operators system, since the space of anti-Hermitian operators is closed with respect to the commutator. Thus, we define the anti-Hermitian operators $H=\mathrm{i}\hat{H}$ and $S=\mathrm{i}\hat{S}$, which realize the unitary \emph{representation} of the abelian Lie algebra $\mathfrak{g}\cong \mathbb{R}^{2}$ in the Hilbert space $V$. The deformations of such a system are essentially deformations of this representation, which is described by the Chevalley-Elienberg complex \cite{Chevalley:1948}.
\subsection{Deformation of the representation and the CE complex}
Let $\mathfrak{g}$ be a Lie algebra with representation $\rho: \mathfrak{g}\rightarrow \mathrm{End}(V)$ in the complex vector space $V$.
\begin{definition}
    Consider the basis $e_i$ in $\mathfrak{g}$, dual basis $c^i$, and the structure constants $f_{ij}^k$. Then the CE complex of $\mathfrak{g}$ with coefficients in $V$ is given by
    \begin{equation}
        \label{CE_complex}
        CE^{\bullet}(\mathfrak{g},V) = S^{\bullet}\mathfrak{g}^*[-1]\otimes V,
    \end{equation}
    with the differential
    \begin{equation}
    \label{CE_diff}
        d =c^i\rho(e_i)-\frac{1}{2}f_{ij}^k c^i c^j\frac{\partial}{\partial c^k}.
    \end{equation}
    Hereafter, the summation convention for repeated indices is assumed.
\end{definition}
\begin{definition}
    Along with the complex \eqref{CE_complex}, we will introduce the Chevalley-Eilenberg complex with values in the representation $\mathrm{End}(V)$, whose differential is defined as the adjoint differential $\{d,\cdot \}$.
\end{definition}
\begin{proposition} \label{prop:first-h}
    The first cohomology group $H^1(\mathfrak{g},\mathrm{End}(V))$ corresponds to the nontrivial infinitesimal deformations of a Lie algebra representation $\rho$.
\end{proposition}
\begin{proposition} \label{prop:second-h}
    The second cohomology group $H^2(\mathfrak{g},\mathrm{End}(V))$ corresponds to the obstructions to  infinitesimal deformations of a Lie algebra representation $\rho$.
\end{proposition}
\begin{example}
    Consider the Lie algebra $\mathfrak{g}=\mathfrak{sl}(2,\mathbb{C})$ with the standard basis elements $e,h,f$ satisfying the following commutation relations:
    \begin{equation}
        [h,e]=2e;\ \ \ [h,f]=-2f,\ \ \ [e,f]=h.
    \end{equation}
    Let $V=\mathbb{C}[x]$ be the space of polynomials in a variable $x$. We consider the representation $\rho:\mathfrak{g}\rightarrow \mathrm{End}(V)$ corresponding to the Verma module with the highest weight $
    \lambda \in \mathbb{C}$. The action of the basis elements is given by the differential operators:
    \begin{equation}
            \rho(e)=\partial_x;\ \ \ 
            \rho(h)=-2x\partial_x+\lambda;\ \ \ 
            \rho(f)=-x^2\partial_x+\lambda x.
    \end{equation}
    For this representation, the Chevalley-Eilenberg differential $d$ can be written as
    \begin{equation}
        \begin{aligned}
            d=c^h\partial_x+c^e(-2x\partial_x+\lambda)+c^f(-x^2\partial_x+\lambda x)\,-\\
            -\,2c^h c^e \frac{\partial}{\partial c^e}+2 c^h c^f \frac{\partial}{\partial c^f}-c^e c^f \frac{\partial}{\partial c^h}.
        \end{aligned}
    \end{equation}
    We analyze an infinitesimal deformation $\rho_t=\rho+t\,\delta\rho$ of this representation:
    \begin{equation}
            \delta\rho(e)=0;\ \ \ \delta\rho(h)=1;\ \ \ \delta\rho(f)=x.
    \end{equation}
    According to Proposition~\ref{prop:first-h}, for $\delta\rho$ to define an infinitesimal deformation, it must correspond to a $1$-cocycle $c^h\delta\rho(h)+c^e\delta\rho(e)+c^f\delta\rho(f)$ in the CE complex:
    \begin{equation}
        d(c^h\delta\rho(h)+c^e\delta\rho(e)+c^f\delta\rho(f))=0.
    \end{equation}
    We verify this condition by performing an explicit calculation:
    \begin{equation}
        d(c^h+c^f x)=-c^e c^f-2c^h c^f[x\partial_x,x]+c^ec^f[\partial_x,x]+2c^hc^fx=0.
    \end{equation}
    Since this deformation corresponds to the shift of the highest weight by $t$, it is not obstructed and defined for any finite $t$, not just infinitesimally.
\end{example}
\subsection{Application to quantum mechanics}
Let $\mathfrak g\cong\mathbb R^2$ be the two-dimensional abelian Lie algebra with basis $(e_1,e_2)$. Let $V$ be a complex vector space equipped with a positive-definite Hermitian inner product
\begin{equation}
    \braket{\cdot,\cdot}:\overline{V}\otimes V\rightarrow \mathrm{End}(V).
\end{equation}
We denote by $\mathfrak u(V)$ the real Lie algebra
\begin{equation}
    \mathfrak u(V)= \left\{
S\in\mathrm{End}(V):
\braket{Sv,w}+\braket{v,Sw}=0
\text{ for all }v,w\in V\right\}.
\end{equation}
Let $\rho:\mathfrak g\to\mathfrak u(V)$ be a unitary representation, and set $H=\rho(e_1)$ and $S=\rho(e_2)$. Thus, $H$ and $S$ are commuting anti-Hermitian operators. All formal deformations considered below have coefficients in $\mathfrak u(V)$.

We construct the Chevalley–Eilenberg complex with coefficients in $\mathfrak u(V)$, equipped with the adjoint action induced by $\rho$. Let $c^H,c^S$ be the basis of $\mathfrak g^{*}[-1]$ dual to $e_1,e_2$, respectively. These elements generate the exterior algebra $S^\bullet\mathfrak g^{*}[-1]$, and the complex takes the form
\begin{equation}
\label{CE}
    0\rightarrow \mathfrak{u}(V)\xlongrightarrow{d}(c^H\oplus c^S)\otimes \mathfrak{u}(V) \xlongrightarrow{d} c^H c^S\otimes \mathfrak{u}(V)\rightarrow 0,
\end{equation}
where the differential $d$ is the CE differential \eqref{CE_diff} for the abelian Lie algebra:
\begin{equation}
\label{CE_diffab}
    d=c^{H}\cdot\mathrm{ad}_{H}+c^{S}\cdot\mathrm{ad}_{S}.
\end{equation}

We assume that $V$ is the direct sum of finite-dimensional joint eigenspaces of the operators $H$ and $S$. We write this decomposition as
\begin{equation}
\label{decomp}
    V=\bigoplus_{(a,\alpha)}V_{(a,\alpha)},
\end{equation}
 where for $v\in V_{(a,\alpha)}$, $Hv=\mathrm{i}\lambda_a v$ and $Sv=\mathrm{i}\mu_\alpha v$ with $\lambda_a,\mu_\alpha\in \mathbb{R}$. The subspaces $V_{(a,\alpha)}$ are pairwise orthogonal since they are eigenspaces of anti-Hermitian operators. We denote the orthogonal projector onto space $V_{(a,\alpha)}$ by $\Pi_{V_{(a,\alpha)}}$. 
 
We also introduce the linear operators
\begin{equation}
    \pi_{(a,\alpha)}^{(b,\beta)}(x)=\frac{1}{2}\cdot \Pi_{V_{(a,\alpha)}}x\,\Pi_{V_{(b,\beta)}}+\frac{1}{2}\cdot \Pi_{V_{(b,\beta)}}x\,\Pi_{V_{(a,\alpha)}}.
\end{equation}
We will denote the image of $\pi_{(a,\alpha)}^{(b,\beta)}$ as $\mathcal{B}_{(a,\alpha)}^{(b,\beta)}$. 

\begin{lemma} \label{lemma:decomposition}
The complex $CE^{\bullet}(\mathfrak{g},\mathfrak{u}(V))$ decomposes into direct sum:
    \begin{equation}
    \label{grading}
        CE^{\bullet}(\mathfrak{g},\mathfrak{u}(V))=\mathcal{Z}^{\bullet}\oplus \mathcal{W}^{\bullet},
    \end{equation}
    where the complexes $\mathcal{Z}^{\bullet}$ and $\mathcal{W}^{\bullet}$ defined as follows
    \begin{equation}
    \begin{aligned}
        \mathcal{Z}^{\bullet}=S^{\bullet}\mathfrak{g}^{*}[-1]\otimes \prod_{(a,\alpha)}\mathrm{Im}\left(\pi_{(a,\alpha)}^{(a,\alpha)}\right),\\\mathcal{W}^{\bullet}= S^{\bullet}\mathfrak{g}^{*}[-1]\otimes \bigcap_{(a,\alpha)}\mathrm{Ker}\left(\pi_{(a,\alpha)}^{(a,\alpha)}\right).
    \end{aligned}
    \end{equation}
    \begin{proof}
    Grouping the diagonal and off-diagonal blocks $\mathcal{B}_{(a,\alpha)}^{(b,\beta)}$ yields $\mathfrak{u}(V)=\mathcal{Z}^0\oplus\mathcal{W}^0$. Each operator has only finitely many nonzero blocks per column. To check that this splitting is preserved by the differential \eqref{CE_diffab}, we compute its action on a general element of \eqref{CE}:
        \begin{eqnarray}
            \notag d(w+c^{H}x+c^{S}y+c^{H}c^{S}z)=[c^H H+c^S S, w]+\{c^H H+c^S S,c^{H}x+c^{S}y\}=\\=c^{H}\cdot [H,w]+c^{S}\cdot [S,w]+c^{H}c^{S}\cdot[H,y]-c^{H}c^{S}\cdot[S,x]\,.
        \end{eqnarray}
    Then it is sufficient to show that for any element $\omega\in\mathcal{B}_{(a,\alpha)}^{(b,\beta)}$, it is true that 
        \begin{equation}
            [H,\omega]\in\mathcal{B}_{(a,\alpha)}^{(b,\beta)},\, [S,\omega]\in\mathcal{B}_{(a,\alpha)}^{(b,\beta)}.
        \end{equation}
    We will prove this for $[H,\omega]$. Let there exist such an $\tilde{\omega}$ that $\omega=\pi_{(a,\alpha)}^{(b,\beta)}(\tilde{\omega})$, then
        \begin{equation}
            [H,\omega]=\frac{\mathrm{i}\lambda_{ab}}{2}\cdot \Pi_{V_{(a,\alpha)}}\tilde{\omega}\,\Pi_{V_{(b,\beta)}}-\frac{\mathrm{i}\lambda_{ab}}{2}\cdot \Pi_{V_{(b,\beta)}}\tilde{\omega}\,\Pi_{V_{(a,\alpha)}}\in \mathfrak{u}(V),
        \end{equation}
    where $\lambda_{ab}=\lambda_{a}-\lambda_{b}$. But then it is easy to find its $\pi_{(a,\alpha)}^{(b,\beta)}$-preimage:
        \begin{equation}
        \label{preimage}
            [H,\omega]=\pi_{(a,\alpha)}^{(b,\beta)}(\mathrm{i}\lambda_{ab}\cdot \Pi_{V_{(a,\alpha)}}\tilde{\omega}\,\Pi_{V_{(b,\beta)}}-\mathrm{i}\lambda_{ab}\cdot\Pi_{V_{(b,\beta)}}\tilde{\omega}\,\Pi_{V_{(a,\alpha)}}).
        \end{equation}
        Then it follows that $[H,\omega]\in\mathcal{B}_{(a,\alpha)}^{(b,\beta)}$. For $[S,\omega]$, the proof is similar.
    \end{proof}
\end{lemma}
\begin{lemma} \label{lemma:acyclic}
    The subcomplex $\mathcal{W}^{\bullet}$ is acyclic.
    \begin{proof}
        It suffices to consider cochains valued in $\mathcal B^{(b,\beta)}_{(a,\alpha)}$, where $(a,\alpha)\ne(b,\beta)$. Each coefficient has finitely many nonzero blocks per column, since it maps the direct sum \eqref{decomp} of finite-dimensional spaces into itself. The primitives below preserve anti-Hermiticity and introduce no new nonzero blocks, so the blockwise solutions combine to a cochain with values in $\mathfrak{u}(V)$.
        
        First, consider a $0$-cycle $\omega$:
        \begin{eqnarray}
            \notag d\omega=c^{H}\cdot \left(\frac{\mathrm{i}\lambda_{ab}}{2}\cdot \Pi_{V_{(a,\alpha)}}\tilde{\omega}\,\Pi_{V_{(b,\beta)}}-\frac{\mathrm{i}\lambda_{ab}}{2}\cdot \Pi_{V_{(b,\beta)}}\tilde{\omega}\,\Pi_{V_{(a,\alpha)}}\right)+\\+\,c^{S}\cdot \left(\frac{\mathrm{i}\mu_{\alpha\beta}}{2}\cdot \Pi_{V_{(a,\alpha)}}\tilde{\omega}\,\Pi_{V_{(b,\beta)}}-\frac{\mathrm{i}\mu_{\alpha\beta}}{2}\cdot \Pi_{V_{(b,\beta)}}\tilde{\omega}\,\Pi_{V_{(a,\alpha)}}\right)=0.
        \end{eqnarray}
        However, $\lambda_{ab}$ and $\mu_{\alpha\beta}$ are not equal to zero simultaneously. Suppose that $\lambda_{ab}\neq 0$:
        \begin{equation}
            \Pi_{V_{(a,\alpha)}}\tilde{\omega}\,\Pi_{V_{(b,\beta)}}- \Pi_{V_{(b,\beta)}}\tilde{\omega}\,\Pi_{V_{(a,\alpha)}}=0.
        \end{equation}
        It follows that $\omega=0$ and, therefore,
        \begin{equation}
        H^{0}(\mathcal{W}^{\bullet},d|_{\mathcal{W}^{\bullet}})=0.
        \end{equation}

 Now, let us consider the $1$-cocycle $\omega=c^{H}\omega_{H}+c^{S}\omega_{S}$:
        \begin{equation}
            d\omega=c^{H}c^{S}[H,\omega_{S}]-c^{H}c^{S}[S,\omega_{H}]=0.
        \end{equation}
        Note that again $\lambda_{ab}$ and $\mu_{\alpha\beta}$ cannot both be zero at the same time. Therefore
        \begin{eqnarray}
            \notag\frac{\mathrm{i}\lambda_{ab}}{2}\cdot \Pi_{V_{(a,\alpha)}}\tilde{\omega}_{S}\,\Pi_{V_{(b,\beta)}}-\frac{\mathrm{i}\lambda_{ab}}{2}\cdot \Pi_{V_{(b,\beta)}}\tilde{\omega}_{S}\,\Pi_{V_{(a,\alpha)}}=\\=\,\frac{\mathrm{i}\mu_{\alpha\beta}}{2}\cdot \Pi_{V_{(a,\alpha)}}\tilde{\omega}_{H}\,\Pi_{V_{(b,\beta)}}-\frac{\mathrm{i}\mu_{\alpha\beta}}{2}\cdot \Pi_{V_{(b,\beta)}}\tilde{\omega}_{H}\,\Pi_{V_{(a,\alpha)}}.
        \end{eqnarray}
        Suppose that $\lambda_{ab}\neq 0$. Then it follows from this equation that 
        \begin{equation}
            \omega_{S}=\frac{\mu_{\alpha\beta}}{\lambda_{ab}}\cdot \omega_{H}.
        \end{equation}
        Then we should find such a $0$-chain $\Omega$ that
        \begin{eqnarray}
            \notag d\Omega=\omega=c^{H}\omega_{H}+\frac{\mu_{\alpha\beta}}{\lambda_{ab}}\cdot c^{S}\omega_{H}=\\\notag=c^{H}\cdot \left(\frac{1}{2}\cdot \Pi_{V_{(a,\alpha)}}\tilde{\omega}_{H}\,\Pi_{V_{(b,\beta)}}+\frac{1}{2}\cdot \Pi_{V_{(b,\beta)}}\tilde{\omega}_{H}\,\Pi_{V_{(a,\alpha)}}\right)+\\ +\,\frac{\mu_{\alpha\beta}}{\lambda_{ab}}\cdot c^{S}\cdot \left(\frac{1}{2}\cdot \Pi_{V_{(a,\alpha)}}\tilde{\omega}_{H}\,\Pi_{V_{(b,\beta)}}+\frac{1}{2}\cdot \Pi_{V_{(b,\beta)}}\tilde{\omega}_{H}\,\Pi_{V_{(a,\alpha)}}\right).
        \end{eqnarray}
        Then $\Omega$ can be chosen as follows:
        \begin{equation}
            \Omega = \frac{1}{2\mathrm{i}\lambda_{ab}}\cdot \Pi_{V_{(a,\alpha)}}\tilde{\omega}_{H}\,\Pi_{V_{(b,\beta)}}-\frac{1}{2\mathrm{i}\lambda_{ab}}\cdot \Pi_{V_{(b,\beta)}}\tilde{\omega}_{H}\,\Pi_{V_{(a,\alpha)}}.
        \end{equation}
        The fact that $d\Omega=\omega$ can be verified by direct calculation. The fact that the cochain $\Omega\in\mathcal{B}_{(a,\alpha)}^{(b,\beta)}$ was shown in \eqref{preimage}. Then we obtained that
        \begin{equation}
        H^{1}(\mathcal{W}^{\bullet},d|_{\mathcal{W}^{\bullet}})=0.
        \end{equation}

        Finally, it remains to show that every $2$-cochain $\omega$ has a preimage. To do this, we will again use the fact that $\lambda_{ab}$ and $\mu_{\alpha\beta}$ are not equal to zero at the same time. Let $\lambda_{ab}\neq 0$ again, then as a $1$-cochain, we will consider
        \begin{equation}
            \Omega=c^{S}\Omega_{S}=c^{S}\cdot\left(\frac{1}{2\mathrm{i}\lambda_{ab}}\cdot \Pi_{V_{(a,\alpha)}}\tilde{\omega}\,\Pi_{V_{(b,\beta)}}-\frac{1}{2\mathrm{i}\lambda_{ab}}\cdot \Pi_{V_{(b,\beta)}}\tilde{\omega}\,\Pi_{V_{(a,\alpha)}}\right).
        \end{equation}
        Then $d\Omega=c^{H}c^{S}[H,\Omega_{S}]$ and using the fact that $\Omega_{S}\in\mathcal{B}_{(a,\alpha)}^{(b,\beta)}$ we obtain
        \begin{equation}
            d\Omega=c^{H}c^{S}\cdot\left(\frac{1}{2}\cdot \Pi_{V_{(a,\alpha)}}\tilde{\omega}\,\Pi_{V_{(b,\beta)}}+\frac{1}{2}\cdot \Pi_{V_{(b,\beta)}}\tilde{\omega}\,\Pi_{V_{(a,\alpha)}}\right)=\omega.
        \end{equation}
        It finally follows from this that $H^{2}(\mathcal{W}^{\bullet},d|_{\mathcal{W}^{\bullet}})=0$.
        \end{proof}
\end{lemma}
\begin{lemma} \label{prop:zero-diff}
    Consider the complex $\mathcal{Z}^{\bullet}$ with differential $d|_{\mathcal{Z}^{\bullet}}$. Then $d|_{\mathcal{Z}^{\bullet}}=0$.
    \begin{proof}
        This immediately follows from the fact that for every $\mathcal{B}_{(a,\alpha)}^{(a,\alpha)}$
        \begin{equation}
            \lambda_{aa}=\mu_{\alpha\alpha}=0.
        \end{equation}
        Therefore, the complex becomes a sequence of vector spaces with zero maps:
        \begin{equation}
            0\rightarrow\mathcal{Z}\xlongrightarrow{0}(c^H\oplus c^S)\otimes\mathcal{Z}\xlongrightarrow{0}(c^H c^S)\otimes\mathcal{Z}\rightarrow 0.
        \end{equation}
        The cohomology groups are simply the vector spaces themselves.
    \end{proof}
\end{lemma}
\begin{theorem} \label{thm:cohomologies}
    The Chevalley-Eilenberg cohomology of the two-dimensional abelian Lie algebra $\mathfrak{g}\cong\mathbb{R}^2$ with coefficients in the module $\mathfrak{u}(V)$ is given by:
    \begin{eqnarray}
        \notag
        H^0(\mathbb{R}^2, \mathfrak{u}(V)) \cong \mathcal{Z}; \\
        \notag
            H^1(\mathbb{R}^2, \mathfrak{u}(V)) \cong \mathcal{Z} \oplus \mathcal{Z}; \\
            H^2(\mathbb{R}^2, \mathfrak{u}(V)) \cong \mathcal{Z},
    \end{eqnarray}
    where $\mathcal{Z}=\{A\in\mathfrak{u}(V)\ |\  [H,A]=[S,A]=0\}$ is the commutant of the representation.
    \begin{proof}
        This immediately follows from the lemmas~\ref{lemma:decomposition}, ~\ref{lemma:acyclic} and ~\ref{prop:zero-diff}.
    \end{proof}
\end{theorem}
\begin{corollary}
    Any $1$-cocycle $c^{H}\delta^{(1)}H+c^{S}\delta^{(1)}S$ is cohomologically equivalent to a cocycle $c^{H}\delta^{(1)}\tilde{H}+c^{S}\delta^{(1)}\tilde{S}$ satisfying the relations
    \begin{equation}
    \begin{aligned}
    \label{1-cocycle_advanced_eq}
        [H, \delta^{(1)}\tilde{S}]=0,\quad [S,\delta^{(1)}\tilde{H}]=0.
    \end{aligned}
    \end{equation}
    Thus, all infinitesimal deformations are exhausted by the solutions to \eqref{1-cocycle_advanced_eq}.
\end{corollary}
\begin{example}
    As an example, let us consider the quantum mechanics of two particles on a circle. That is, we take the state space to be $V = L^2(S^1 \times S^1)$. The elements of this space are square integrable functions $\psi(q_{1},q_{2})$ periodic (with period $2\pi$) in both arguments. Strictly speaking, to remain within the framework of the preceding proofs, we should work with trigonometric polynomials. However, this does not affect the main conclusion of the example. Consider the operators:
    \begin{equation}
    \label{hamiltonians}
    \begin{aligned}
        \hat{H}=p_{1}^{2}+p^{2}_{2},\quad \hat{S}=p_{1}+p_{2},
    \end{aligned}
    \end{equation}
    where $p_{i}=-\mathrm{i}\frac{\partial}{\partial q_{i}}$ are the momentum operators. Clearly, these operators commute with each other. The eigenfunctions of $\hat{H}$ and $\hat{S}$ have the form
    \begin{equation}
        \psi_{n_{1},n_{2}}(q_{1},q_{2})=\frac{1}{2\pi}\,e^{\mathrm{i}n_{1}q_{1}+\mathrm{i}n_{2}q_{2}}.
    \end{equation}
    As a 1-cocycle deforming this system, let us consider
    \begin{equation}
    \label{deformations}
        \delta^{(1)}\hat{H}=\sin(q_{1}-q_{2}),\quad \delta^{(1)}\hat{S}=p_{1}-p_{2}.
    \end{equation}
    Let us try to extend this deformation to the second order of perturbation theory:
    \begin{equation}
    \begin{aligned}
    \label{example_eq}
        [\hat{H},\delta^{(2)}\hat{S}]+[\delta^{(2)}\hat{H},\hat{S}]+\underbrace{2\mathrm{i}\cos(q_{1}-q_{2})}_{[\delta^{(1)}\hat{H},\,\delta^{(1)}\hat{S}]}=0.
    \end{aligned}
    \end{equation}
    Now we will show that this equation is unsolvable using the logic from the proof of theorem~\ref{thm:cohomologies}. Namely, let us consider the subspace spanned by the wavefunctions $\psi_{0,1}, \psi_{0,-1}, \psi_{1,0}$, and $\psi_{-1,0}$. This is an eigenspace of $\hat{H}$ with eigenvalue $1$. The matrix of the operator $\hat{S}$ in this basis has the form $\mathrm{diag}(1,-1,1,-1)$. Acting with equation \eqref{example_eq} on the function $\psi_{0,1}$ and computing the inner product with $\psi_{1,0}$. Since these functions are eigenfunctions of $\hat{H}$ and $\hat{S}$ with the same eigenvalues, the first two terms vanish. The third term gives
    \begin{equation}
        \frac{2\mathrm{i}}{4\pi^{2}}\int_{0}^{2\pi}\mathrm{d}q_{1}\int_{0}^{2\pi}\mathrm{d}q_{2}\,e^{\mathrm{i}(q_{2}-q_{1})}\cos(q_{1}-q_{2})=\mathrm{i}.
    \end{equation}
    Therefore, equation \eqref{example_eq} is unsolvable. Thus, in this example, the $2$-cocycle
    \begin{equation}
        c^{H}c^{S}[\delta^{(1)}\hat{H},\delta^{(1)}\hat{S}]\,
    \end{equation}
    has a non-zero cohomology class. This means that the deformation of the system of Hamiltonians \eqref{hamiltonians} cannot be extended along the direction specified by \eqref{deformations}.
\end{example}
\section{Graded Lie algebra structure on $H^{\bullet}(\mathbb{R}^{2},\mathfrak{u}(V))$}
\begin{proposition}
    \label{lemma:no_obs}
    There are no higher orders obstructions. The $L_{\infty}$-algebra structure on the cohomology $H^{\bullet}(\mathbb{R}^{2}, \mathfrak{u}(V))$ reduces to a differential graded Lie algebra (dgLa) structure with a vanishing differential, meaning that all higher obstruction maps $\mu_{n}$ vanish for $n \geq 3$.
    \begin{proof}
    According to Theorem ~\ref{thm:cohomologies}, the Chevalley-Eilenberg cohomology groups are isomorphic to the subcomplex $\mathcal{Z}^{\bullet}$. As established in Lemma \ref{prop:zero-diff}, the differential $d$ vanishes identically on this subcomplex ($d|_{\mathcal{Z}^{\bullet}} = 0$) because the corresponding eigenvalues satisfy $\lambda_{aa} = \mu_{\alpha\alpha} = 0$.
    
    In the context of deformation theory, the $L_{\infty}$-structure on cohomology is induced from the structure of the initial complex via the homotopy transfer theorem \cite{Arvanitakis:2020homotopy}. Since the complex decomposes into a direct sum of an acyclic part and a part with zero differential $\mathcal{Z}^{\bullet}$, we can choose a homotopy transfer where the only non-trivial higher bracket is the binary commutator inherited from $\mathfrak{u}(V)$. 
    
    Specifically, the obstruction maps $\mu_n$ for $n \geq 3$ involve higher $l_n$ brackets which vanish when the differential on the representative subcomplex is zero and no further corrections to the generators are required beyond the second order. Thus, the only possible obstruction to the deformation $(\hat{H}(t), \hat{S}(t))$ arises in the second order.
    \end{proof}
\end{proposition}
\begin{remark}
The Proposition~\ref{lemma:no_obs} has an interesting physical consequence: if the original Hamiltonian $\hat{H}$ or the symmetry generator $\hat{S}$ had no degenerate eigenvalues, the anomaly upon perturbation of such a system would be zero.
\end{remark}
\begin{proposition}
    The only perturbative anomaly in quantum mechanics is realized in the second order of perturbation theory and satisfies a quadratic equation.
    \begin{proof}
        As shown earlier, the deformation space is the commutant of the representation of the original algebra. It is isomorphic to the direct sum of the A-type Lie algebras $\prod_{(a,\alpha)}\mathfrak{u}(n_{(a,\alpha)})$, where $n_{(a,\alpha)}=\mathrm{dim}\,V_{(a,\alpha)}$. So we can decompose $\delta^{(1)} H$ and $\delta^{(1)} S$ into a basis $e_{i}$ of the commutant $\mathcal{Z}$: $\delta^{(1)} H=h^{i}e_{i}, \delta^{(1)} S=s^{i}e_{i}$. We will denote by $f_{ij}^{k}$ the structure constants of the commutant $\mathcal{Z}$ in the basis $e_{i}$. Then, deforming along the basis directions in the first order, we obtain the following equation in the second order of perturbation theory:
        \begin{equation}
            [\delta^{(2)}S,H]+[S,\delta^{(2)}H]=h^i s^j f_{ij}^{k}\,e_{k}.
        \end{equation}
        The obstruction is proportional to $f_{ij}^{k}$. It has a quadratic equation that comes from the dgLa structure on $H^{\bullet}(\mathbb{R}^{2},\mathfrak{u}(V))$, which is essentially the Jacobi identity:
        \begin{equation}
            f_{ij}^{l} f_{lk}^{m} + f_{jk}^{l} f_{li}^{m} + f_{ki}^{l} f_{lj}^{m} = 0.
        \end{equation}
        Therefore, the anomaly in this example is limited to the second order of perturbation theory and, according to general ideology, satisfies a quadratic equation.
     \end{proof}
\end{proposition}
\section{Conclusion}
In this paper, we examine the simplest example of an anomaly that arises as a cohomological obstruction. Next, we plan to implement the cohomological program in a crucial example: conformal symmetry in two-dimensional quantum field theory.

Another interesting special case is where the anomaly arises due to deformation along the Planck constant $\hbar$. Then we can consider a quantum mechanical example in which deformation occurs simultaneously in two directions: the Planck constant $\hbar$ is deformed and the coupling constant $t$ is deformed. The deformation of the coupling constant $t$ at $\hbar=0$ corresponds to a transition to an action that is classically invariant with respect to the deformed symmetry. On the other hand, the deformation of the Planck constant $\hbar$ at $t=0$ corresponds to the existence of quantum symmetry in the initial theory. Then, simultaneous deformation along these two directions can lead to a \emph{mixed-type} anomaly (with parameter t$\hbar $), which we interpret as a one-loop anomaly. We intend to construct an example of this phenomenon in one of our next papers.

\section*{Data availability statement}
Data availability is not applicable to this article as no data were created or analyzed during this study. All theoretical results and relevant calculations are included within the published text.

\section*{Conflict of interest statement}
On behalf of all authors, the corresponding author states that there is no conflict of interest.

\section*{Acknowledgements}
We are grateful to Vyacheslav Lysov for helpful discussions. We are grateful to the Shanghai Institute for Mathematics and Interdisciplinary Sciences for providing us a workspace. The first author is supported by the Ministry of Science and Higher Education of the Russian Federation (Agreement No. 075-15-2025-013).

\bibliographystyle{alphaurl}
\bibliography{refs.bib}

\end{document}